\RequirePackage[l2tabu, orthodox]{nag}
\documentclass[10pt]{article}

\usepackage[T1]{fontenc}
\usepackage{lmodern}
\usepackage[square, numbers]{natbib}
\usepackage{microtype}
\DeclareFontSeriesDefault{bf}{bx}          
\usepackage[frenchmath]{mathastext}		   
\usepackage{amsmath}                       
\numberwithin{equation}{section}
\usepackage{amsfonts}
\usepackage{graphicx}                      
\usepackage[plainpages=false, pdfpagelabels]{hyperref} 
	\hypersetup{
		colorlinks   = true,
		citecolor    = RoyalBlue, 
		linkcolor    = RubineRed, 
		urlcolor     = Turquoise
	}
\usepackage[dvipsnames]{xcolor}
\usepackage[all]{xy}
\usepackage{amssymb}                        
\usepackage[mathscr]{eucal}                 
\usepackage{dsfont}							
\usepackage[paperwidth=8.5in,paperheight=11.00in,top=1.25in, bottom=1.25in, left=1.00in, right=1.00in]{geometry}
\usepackage{mathtools}                      
\mathtoolsset{showonlyrefs=true}            
\usepackage{fixltx2e,amsmath}               
\MakeRobust{\eqref}
\usepackage{microtype}		                
\usepackage{bbm}                            
\usepackage{physics}                        
\usepackage{enumerate}                      
\usepackage{subcaption}                     
\usepackage{booktabs}                       
\usepackage{amsthm}                         
\allowdisplaybreaks                         
\theoremstyle{plain}
\newtheorem{theorem}{Theorem}
\numberwithin{theorem}{section}
\newtheorem{lemma}[theorem]{Lemma}          
\newtheorem{proposition}[theorem]{Proposition}

\theoremstyle{definition}

\newtheorem{remark}[theorem]{Remark}
\newtheorem{assumption}[theorem]{Assumption}

\newcommand{\pd}{\partial}
\newcommand{\one}{\mathbbm{1}}

\newcommand{\E}{\mathbb{E}}
\newcommand{\N}{\mathbb{N}}
\newcommand{\Q}{\mathbb{Q}}
\newcommand{\Pb}{\mathbb{P}}
\newcommand{\R}{\mathbb{R}}

\newcommand{\cD}{\mathcal{D}}
\newcommand{\cK}{\mathcal{K}}

\newcommand{\cT}{\mathcal{T}}

\newcommand{\what}{\widehat}

\renewcommand{\(}{\left(}
\renewcommand{\)}{\right)}
\renewcommand{\[}{\left[}
\renewcommand{\]}{\right]}

\DeclareMathOperator*{\argmin}{arg\,min}

\author{
    Jimin Lin\thanks{Quantitative Research, Bloomberg. (E-mail: \href{mailto:jlin846@bloomberg.net}{jlin846@bloomberg.net}).}
    \and Guixin Liu\thanks{Quantitative Research, Bloomberg. (E-mail: \href{mailto:gliu230@bloomberg.net}{gliu230@bloomberg.net}). } 
}

\begin{document}
\title{Neural Term Structure of Additive Process for Option Pricing}
\date{}
\maketitle

\begin{abstract}
The additive process generalizes the L\'evy process by relaxing its assumption of time-homogeneous increments and hence covers a larger family of stochastic processes. Recent research in option pricing shows that modeling the underlying log price with an additive process has advantages in easier construction of the risk-neural measure, an explicit option pricing formula and characteristic function, and more flexibility to fit the implied volatility surface. Still, the challenge of calibrating an additive model arises from its time-dependent parameterization, for which one has to prescribe parametric functions for the term structure. For this, we propose the neural term structure model to utilize feedforward neural networks to represent the term structure, which alleviates the difficulty of designing parametric functions and thus attenuates the misspecification risk. Numerical studies with S\&P 500 option data are conducted to evaluate the performance of the neural term structure.
\end{abstract}

\section{Introduction}
Providing an arbitrage-free valuation formula and specifying risk-neutral dynamics are essentially two sides of the same coin in option pricing. Yet, the modeling methodology has been leaning towards the latter for decades. That is, the invention of an option pricing model typically starts with proposing a stochastic process that is a martingale for the underlying asset, so that the corresponding risk-neural measure is constructed, and henceforth the arbitrage-free option valuation can be determined either analytically or numerically. Such a methodology was established through the pioneering work of \citet{bachelier1900theorie} and \citet{black1973pricing}, and since then, almost all of the prevailing models have been invented along this paradigm. The list includes but is not limited to local volatility models by \citet{dupire1994pricing, cox1996notes}, stochastic volatility models by \citet{heston1993closed, hagan2002managing, bates1996jumps}, jump-diffusion models by \citet{merton1976option, kou2002jump}, and other models built upon L\'evy processes by \citet{madan1998variance, barndorff1997normal}.

Nonetheless, the reverse approach, which first provides an arbitrage-free valuation formula, as in \citet{carr2005note, davis2007range}, and then finds the underlying martingale supporting the formula, is still possible, as noted in \cite{kellerer1972markov, madan2002making}. In recent work, \citet{carr2021additive} starts with one particular pricing formula that yields logistically distributed marginals. Although there is no underlying L\'evy process that produces such marginals, by allowing the increment to be nonstationary, an additive logistic process can be constructed to support that pricing formula. Treatment of using additive processes to price and calibrate options can be traced back to \citet{cont2004financial}. \citet{azzone2023explicit} shows that the additive logistic model has desirable qualities for option pricing - a simple closed formula for fast computation, the flexibility of shape control for calibration, and the ease to simulate with fast Fourier transform-based Monte Carlo \cite{azzone2023fast}.

An appealing alternative approach to generate option pricing models is hence revealed by \citet{carr2021additive, azzone2023explicit}. One can choose a candidate density from an infinitely divisible family and generalize it with some skew transformation to access more control over its moments, such that its marginal density can fit the volatility smile first at one single expiry. Then, if technical conditions hold as in \cite[Theorem 9.8]{ken1999levy}, an additive process supporting such marginal can be constructed in a straightforward manner by making the distribution parameters time-dependent, or in other words, prescribing a \textit{term structure} to the distribution parameter. Finally, it is easy to calculate the compensator to make this additive process a martingale.

It is interesting to compare the \cite{carr2021additive, azzone2023explicit} procedure with the way in which \cite{de2019building, conze2021bass, bourgey2024fast} build arbitrage-free volatility surfaces, analogically. Both start by specifying one or a few marginals and then find the martingale measure supporting the marginals, while complexity occurs in different stages of each process. The latter can freely use non-parametric methods to fit marginals, but acquiring a martingale measure involves complicated approaches, such as Markov functional construction with the Sinkhorn optimization algorithm. On the other hand, for the former, it is not easy to find a proper infinitely divisible marginal to produce an additive process, but once found, the martingale is readily obtained. 

Suppose one successfully constructs an additive process with a proper distribution, it remains non-trivial to specify its term structure, i.e., a set of time-dependent, deterministic, and continuous functions such that the measure agrees with the market behaviors. For example, \cite{azzone2023explicit} provides parametric candidates for the term structure based on a detailed analysis that covers short- and long-maturity asymptotics, large strike asymptotics, and behaviors of the model-induced implied volatility.

In this research, we alleviate the difficulty of term structure selection by utilizing the universal approximation property of a feedforward neural network. Constraints on the term structure can be enforced by regulating the neural network behavior with penalty terms. We call the term structure represented by a neural network the \textit{neural term structure}. Throughout this research, we exemplify the usage of neural term structure specifically on the additive logistic process, but this framework can be easily implemented with any other additive process.

Numerical studies are conducted to assess the performance of the neural term structure with additive logistic processes. We first see that the neural term structure can fit the synthetic surface generated by the parametric term structure in \cite{azzone2023explicit}. Then with an S\&P 500 implied volatility surface, we show that the neural term structure better calibrates to the real market data than the chosen parametric term structure. Meanwhile, we notice there may exist multiple dissimilar term structures that can achieve close calibration performance on a single surface. This raises concerns about the robustness of the calibration. For this, we extend the neural term structure to jointly train on a sequence of historical surfaces. It produces a smooth flow of term structures that fits both generic surfaces and specific local dynamics, which is visualized as the \textit{term structure surfaces}.

The remaining content of this paper is organized as follows. Section~\ref{sec:additive} scrutinizes the new option pricing model generating procedure based on the additive process, exemplified with the additive logistic process that is extensively studied in \cite{carr2021additive, azzone2023explicit}. Section~\ref{sec:cali} discusses methods for calibrating the additive process based model and introduces the neural term structure for fitting a single implied volatility surface and a sequence of historical surfaces. Section~\ref{sec:exp} demonstrates the advantage of the neural term structure with numerical experiments. Section~\ref{sec:con} concludes this paper.

\section{Additive Option Pricing Model} \label{sec:additive}

In this section, we first provide a review of the most relevant properties of the generalized logistic (GL) distribution. More details on GL can be found in \cite{balakrishnan1991handbook, barndorff1982normal}. Then, following the work of \citet{carr2021additive, azzone2023explicit}, we explain the construction of the additive logistic process to model the underlying log-price process and the corresponding closed-form option pricing formula.

\subsection{Generalized Logistic Distribution} \label{subsec:gl}

Denote $\text{L}(\mu, \sigma)$ the location-scale logistic distribution, $\mu \in \R$ and $\sigma \in \R_+$. For $x \in \R$, let $f^L$ and $F^L$ be respectively its PDF and CDF given by
\begin{align}
f^L(x; \mu, \sigma) = \frac{e^{-\frac{x-\mu}{\sigma}}}{\sigma \(1 + e^{-\frac{x-\mu}{\sigma}}\)^2}, &&
F^L(x; \mu, \sigma) = \frac{1}{1 + e^{-\frac{x-\mu}{\sigma}}}.
\end{align}
The generalized logistic distribution of type IV, also known as logistic-beta distribution, $\text{GL}(\mu, \sigma, \alpha, \beta)$ with $\alpha, \ \beta \in \R_+$, is generated by applying beta-skew-transformation to $\text{L}(\mu, \sigma)$, whose PDF $f^{GL}$ and CDF $F^{GL}$ are given by
\begin{align}
f^{GL}(x; \mu, \sigma, \alpha, \beta)
    &= \frac{1}{B(\alpha, \beta)} f^L(x; \mu, \sigma) \(F^L(x; \mu, \sigma)\)^{\alpha-1} \(1 - F^L(x; \mu, \sigma)\)^{\beta - 1}, \\
F^{GL}(x; \mu, \sigma, \alpha, \beta)
    &= \frac{1}{B(\alpha, \beta)} \int_0^{F^L(x; \mu, \sigma)} u^{\alpha-1} \(1-u\)^{\beta-1} \dd u.
\end{align}
where $B(\alpha, \beta)=\frac{\Gamma(\alpha)\Gamma(\beta)}{\Gamma(\alpha+\beta)}$ is the beta function and $\Gamma$ is the gamma function. Note that $F^{GL}$ is obtained simply by a change of variable $\dd u = F^L(\dd x; \mu, \sigma)$ to the integral of $f^{GL}$ and can be recognized as the regularized incomplete beta function integrated up to $F^L(x; \mu, \sigma)$. In addition, using the property of the location-scale family, we can pivot the $\text{L}(\mu, \sigma)$ and hence define the standard GL distribution with skew parameters $\alpha$ and $\beta$ by
$$
\Phi(x; \alpha, \beta):= F^{GL}(x; 0, 1, \alpha, \beta),
$$
which is useful to express the option pricing formula more compactly in the later content. 

$\text{GL}(\mu, \sigma, \alpha, \beta)$ has the characteristic function for $z \in \R$ 
\begin{align} \label{eq:char}
\what{f}^{GL}(z; \mu, \sigma, \alpha, \beta) = \(\frac{B(\alpha + i \sigma z, \beta - i \sigma z)}{B(\alpha, \beta)}\) e^{i\mu z}.
\end{align}
The finite divisibility and self-decomposability of $\text{GL}(\mu, \sigma, \alpha, \beta)$ are important to construct an additive logistic process, which can be identified through its characteristic function.
\begin{lemma} \label{lemm:char}
$\text{GL}(\mu, \sigma, \alpha, \beta)$ is infinitely divisible with the L\'evy triplet $(a, 0, v(x)\dd x)$ given by
\begin{align} \label{eq:triplet}
a
    = \sigma \int_{0}^{\sigma^{-1}} \frac{e^{-\beta x} - e^{-\alpha x}}{1 - e^{-x}} \dd x + \mu, &&
v(x)
    = \frac{e^{-\(\beta\one_{x>0} + \alpha\one_{x<0}\)\frac{\abs{x}}{\sigma}}}{\abs{x}\(1-e^{\frac{-\abs{x}}{\sigma}}\)}.
\end{align}
Moreover, it is also self-decomposable.
\end{lemma}
\begin{proof}
Express the beta function in \eqref{eq:char} with the integral representation of the log-gamma function
\begin{align}
\log \Gamma(z) = \int_0^{\infty} \(\frac{e^{-zx}-e^{-x}}{x(1-e^{-x})} - (z-1)\frac{e^{-x}}{x}\) \dd x, && \Re(z) > 0.
\end{align}
The cumulant function, i.e. natural logarithm of the characteristic function \eqref{eq:char} can be organized into the representation
\begin{align}
\log \what{f}^{GL}(z; \mu, \sigma, \alpha, \beta) = i z a - \frac{1}{2} z^2 b + \int_{\R} \(e^{izx} - 1 - iz x\one_{\abs{x}<1}\) v(x) \dd x,
\end{align}
where the triplet $(a, b, v(x) \dd x)$ reads $b=0$, $a$ and $v(x)$ are as in the Lemma~\ref{lemm:char}; hence, by the L\'evy-Khintchine Theorem we have the infinitely divisibility. Notice the criteria $\dd (\abs{x} v(x)) / \dd x$ is positive for $x < 0$ and negative for $x > 0$, the self-decomposability follows by \cite[Theorem~5.11.2]{lukacs1970characteristic}.
\end{proof}

\subsection{Additive Logistic Process}
Based on Lemma~\ref{lemm:char}, we can construct an additive process whose marginal distribution is $\text{GL}$ by simply replacing the constant parameters $\mu, \sigma, \alpha, \beta$ in Section~\ref{subsec:gl} by appropriate continuous functions of time. These time-dependent parameters, named the \textit{term structure}, jointly determine the cumulants, especially the mean, volatility, skewness, and kurtosis of interest, of the entire process. Consequently, we also have the L\'evy triplet in Equation~\eqref{eq:triplet} become time-dependent here, $(a(t), 0, v(t, x) \dd x)$. Assumptions required on the term structure are as follows.
\begin{assumption} \label{ass:term}
$\mu : \R_+ \to \R$ and $\sigma, \alpha, \beta: \R_+ \to \R_+$ are continuous functions and we have
    \begin{enumerate}[$\quad$ i.]
        \item $\frac{\alpha}{\sigma}$ and $\frac{\beta}{\sigma}$ non-increasing, \label{ass:one}
        \item $\sigma$ non-decreasing and $\sigma(0)=0$, \label{ass:two}
        \item $0 < \alpha,\ \beta < \infty$, \label{ass:three}
        \item $\sigma < \beta$. \label{ass:four}
    \end{enumerate}
\end{assumption}
The first three conditions in Assumption~\ref{ass:term} are sufficient to construct the additive logistic process.
\begin{proposition} \label{prop:additive}
For $\mu$, $\sigma$, $\alpha$, and $\beta$ satisfying conditions~\ref{ass:one} - \ref{ass:three} in Assumption~\ref{ass:term}, there exists in law a unique additive process $(X_t)_{t>0}$ starting at zero such that
\begin{align}
    X_t \sim \text{GL}\(\mu(t), \sigma(t), \alpha(t), \beta(t)\).
\end{align}
\end{proposition}
\begin{proof}
It is easy to verify that both $a(t)$ and $v(t, \cdot)$ are continuous for $t \ge 0$, $a(0)=0$, $v(0, \cdot)=0$, and $v(s, \cdot) \le v(t, \cdot)$ for $0 \le s \le t$, see \cite{carr2021additive} for details. The result follows by applying the second part of \cite[Theorem~9.8]{ken1999levy}.
\end{proof}
Notice from the $\text{GL}$ characteristic function \eqref{eq:char} that we will need condition \ref{ass:four} in Assumption~\ref{ass:term} for $X_t$ to have a finite expectation. With one more extra condition on $\mu$, we can hitherto obtain an additive logistic martingale.
\begin{proposition}\label{prop:additive_mart}
Assume $r > 0$ is the constant risk-free rate, $\sigma$, $\alpha$, and $\beta$ satisfy Assumption~\ref{ass:term}, and
\begin{align} \label{eq:mu}
\mu(t) = \log\left(\frac{B(\alpha(t) + \sigma(t), \beta(t) - \sigma(t))}{B(\alpha(t), \beta(t))}\right),
\end{align}
then there exists in law a unique additive logistic process $(X_t)_{t\ge 0}$
\begin{align} \label{eq:additive}
X_t \sim \text{GL}\(rt - \mu(t), \sigma(t), \alpha(t), \beta(t)\),
&& X_0 = 0,
\end{align}
whose discounted exponential $e^{-r t} S_t = e^{-r t + X_t}$ is a $\Q$-martingale.
\end{proposition}
\begin{proof}
Existence follows from Proposition~\ref{prop:additive}. For martingality, inserting Equation~\eqref{eq:mu} to Equation~\eqref{eq:char} we have $e^{-r t}\E[S_t] = \varphi_t(-i) = 1$, where $\varphi_t(\cdot)=\hat{f}^{GL}(\cdot \ ; rt - \mu(t), \sigma(t), \alpha(t), \beta(t))$.
\end{proof}
\begin{remark} \label{rema:stationary}
Note that the stationary increment assumption, which we are used to in L\'evy processes, is dropped for additive processes, i.e. for all $0 \le s < t$, $X_t - X_s \simeq X_{t-s}$ does not necessarily hold. With additive processes, one should be careful of any methods based on time-homogeneity property. For instance, $X_t$ cannot be simulated with conventional Monte Carlo i.i.d. sampling. Nonetheless, we still have the analytic characteristic function for increment $\E[e^{iz(X_t - X_s)}] = \varphi_t / \varphi_s$ that can be used to simulate $X_t$ by fast Fourier transform \cite{azzone2023fast}.
\end{remark}

\subsection{Option Pricing Formula}
Assume the price process of the underlying asset follows $S_t = S_0 e^{X_t}$ where $S_0$ is the spot price and $X_t$ is an additive logistic martingale as in Proposition~\ref{prop:additive_mart}. The explicit marginal density of $X_t$ leads to the explicit pricing formulas for the initial value of options.
\begin{proposition}
\label{prop:option}
At time $t=0$, the initial value of an European call $C(K, T)$ and put $P(K, T)$ with strike $K$ and time-to-maturity $T$ is priced by
\begin{equation}\label{eq:option}
\begin{aligned}
C(K, T) 
    &= S_0 \Phi\(d; \beta(T) - \sigma(T), \alpha(T) + \sigma(T)\) - e^{-r T} K \Phi\(d; \beta(T), \alpha(T)\), \\
P(K, T)
    &= e^{-r T} K \Phi\(-d; \alpha(T), \beta(T)\) - S_0 \Phi\(-d; \alpha(T) + \sigma(T), \beta(T) - \sigma(T)\),
\end{aligned}
\end{equation}
where
$$
d = \frac{\log\frac{S_0}{K} + r T - \mu(T) }{\sigma(T)}.
$$
\end{proposition}
\begin{proof}
Pivoting the log-return $X_T \simeq rT - \mu(T) + \sigma(T) Z$ with a standard GL random variable $Z \sim \text{GL}(0, 1, \alpha(T), \beta(T))$ we have
\begin{align}
P(K,T)
    &= e^{-rT}\E\[\(K - S_T\)^+\] = e^{-rT}\E\[\(K - S_0 e^{X_T}\)\one_{Z \le -d}\] \\
    &=  e^{-rT} K \Pb\(Z \le -d\) - S_0 \E\[e^{Z}\one_{Z \le -d}\].
\end{align}
Notice that $\E\[e^{Z}\one_{Z \le -d}\] = \Pb(\tilde{Z}\le -d)$ with $\tilde{Z} \sim \text{GL}(0, 1, \alpha(T) + \sigma(T), \beta(T) - \sigma(T))$, and reexpress $\Pb$ by $\Phi$, we obtain the result. $C(K, T)$ follows by put-call parity and beta reflection symmetry.
\end{proof}
\begin{remark}\label{rema:option}
The pricing formulas \eqref{eq:option} are analogous to, and yield the same simplicity and computability as the Black-Scholes formula with standard GL CDFs with prescribed term structure $(\sigma, \alpha, \beta)$ substituting the standard normal CDF. However, it is not directly applicable to price the option price process, say $C_t(K, T)$ conditioning in $S_t$ for $t > 0$, due to the time-inhomogeneity, as noted in Remark~\ref{rema:stationary}.
\end{remark}

\section{Calibrating additive model with neural term structure} \label{sec:cali}
With the additive logistic martingale \eqref{eq:additive} constructed, it remains to specify functional forms for term structure $(\sigma, \alpha, \beta) \in \Xi$ to calibrate the model to the implied volatility surface. Selecting the term structure is non-trivial and subject to the modeler's empirical experience to translate between the surface characteristic and the proper functional forms, and it is impractical to manually configure the functions by trial and error. A uniform framework to find the optimal term structure is desired.

For this, we design the machine learning-based approach by utilizing a feedforward neural network that is regulated by Assumption~\ref{ass:term} to represent the term structure, named the \textit{neural term structure}. The neural term structure is not only capable of calibrating to a single volatility surface, but is also scalable to fit a long historical sequence of implied volatility surfaces to capture the surface dynamic. In this section, we begin with a single surface calibration case, and then we will see that it is effortless to extend the method for joint calibration on surface dynamics.

\subsection{Calibration with a single surface} \label{subsec:cali_one}
One implied volatility represents the information of an option chain, i.e., a set of market options listed by a grid of combinations of strikes and maturities, that are quoted at one time. Let $\cK = \{K_1, \dots, K_n\}$ and $\cT = \{T_1, \dots, T_m\}$ respectively collect $n \in \N$ strikes and $m \in \N$ maturity dates, let each $(K, T) \in \cK \times \cT$, $C^M(K, T)$ be the market price of a call option with strike $K$ and maturity $T$, and let $C(K, T; \sigma(T), \alpha(T), \beta(T))$ be the call price given by formula \eqref{eq:option} prescribed by the term structure $(\sigma, \alpha, \beta)$. Let $\Xi$ be the space of all functional forms of term structure that satisfy the Assumption~\ref{ass:term}. The calibration task is to minimize the pricing error
\begin{align} \label{eq:obj_p}
\min_{(\sigma, \alpha, \beta) \in \Xi} L_P(\sigma, \alpha, \beta).
\end{align}
Based on specific calibration purposes, different metrics and weighting schemes can be applied to $L_P$, such as bid-ask spread weighted mean square error and Black-Scholes vega weighted mean square error in \cite{hamida2005recovering}. Here we simply use the plain mean squared error
\begin{align}
L_P(\sigma, \alpha, \beta)
    = \frac{1}{\abs{\cK}\abs{\cT}}\sum_{\substack{K\in \cK \\ T \in \cT}} \(C^M(K, T) - C(K, T; \sigma(T), \alpha(T), \beta(T))\)^2.
\end{align}

There are three approaches to minimizing $L_P$.

\textbf{Calibration on slices of marginals}: naively, one can fit marginals slice by slice with scalar parameters, and then interpolate the calibrated parameters between slices to construct the term structure. Ideally, such a term structure is feasible, otherwise ad-hoc adjustment is needed. That is, for every fixed maturity $T_i \in \cT$ $(1 \le i \le m)$, one finds $(\sigma_{T_i}, \alpha_{T_i}, \beta_{T_i}) \in \R_+^3$ such that
\begin{align}
(\hat{\sigma}_{T_i}, \hat{\alpha}_{T_i}, \hat{\beta}_{T_i}) = \argmin_{\sigma, \alpha, \beta} \frac{1}{\abs{\cK}} \sum_{K\in \cK} \(C^M(K, T_i) - C(K, T_i; \sigma, \alpha, \beta)\)^2.
\end{align}

Next, the term structure $(\hat{\sigma}, \hat{\alpha}, \hat{\beta})$ is constructed by interpolating the time slices $(\hat{\sigma}(T_i), \hat{\alpha}(T_i), \hat{\beta}(T_i)) = (\hat{\sigma}_{T_i}, \hat{\alpha}_{T_i}, \hat{\beta}_{T_i})$ for $1 \le i \le m$, with either lines, polynomials or splines. The interpolated term structure is feasible if $(\hat{\sigma}, \hat{\alpha}, \hat{\beta}) \in \Xi$.

However, this approach does not guarantee feasibility. A limited number of samples in each slice could induce noise in parameter estimation, which poses difficulty in obtaining a proper interpolation. Should the infeasibility occur, one still needs to jointly regulate the term structure across the slices, which in turn undermines the straightforwardness of this approach.

\textbf{Calibration with prescribed functions}: one can design parametric functions to fully control the term structure according to its knowledge of the implied volatility surface behavior. This reduces the optimization problem from the functional space $\Xi$ to a smaller space of scalar parameters. \citet{carr2021additive} tests the simplest symmetric standard GL with
\begin{align}
\sigma(\tau) = \left(1 - e^{-\tau \sigma_0^{1/H}}\right)^H,
&& \alpha(\tau)=\beta(\tau) \equiv 1,
\end{align}
where the term structure is parameterized by only two parameters, the volatility level $\sigma_0$ and volatility term exponent $H$.

\citet{azzone2023explicit} provides for each term structure function a simple form and a sophisticated form
\begin{equation} \label{eq:term_func}
\begin{aligned}
&\sigma^a(\tau) = \sigma_0 \tau^{H_0},  
&&\alpha^a(\tau) = \alpha_0, 
&&\beta^a(\tau) = \beta_0 + \sigma^a(\tau) \\
&\sigma^b(\tau) = \frac{\arctan(\sigma_0 \frac{\pi}{2}\tau^{H_0})}{\arctan(\sigma_1^{-1} \frac{\pi}{2}\tau^{-H_1})}, 
&&\alpha^b(\tau) = \alpha_1 + \frac{\alpha_0 - \alpha_1}{1 + \sigma(\tau)},
&&\beta^b(\tau) = \beta_1 + \frac{\beta_0 - \beta_1}{1 + \sigma^b(\tau)} + \sigma^b(\tau),
\end{aligned}
\end{equation}
where $\sigma_0, \sigma_1, \alpha_0, \alpha_1, \beta_0, \beta_1, H_0, H_1$ are parameters. One can choose any combination of them to configure the term structure. The basic term functions can serve as baselines and the sophisticated ones with more parameters have better flexibility to control the shape of the implied volatility surface.

The obvious shortcomings of this approach are that it is not easy to design proper functions and that, even if so, the designed function can hardly be optimal in $\Xi$.

\textbf{Calibration with the neural term structure}: denote an $N \in \N$ layer neural network $\eta$
\begin{align}  \label{eq:eta}
\eta: \R_+ \to \R^3_+,
&&\tau \mapsto \eta(\tau) = \(\sigma^\eta(\tau), \alpha^\eta(\tau), \beta^\eta(\tau)\),
\end{align}
with
\begin{align} \label{eq:nn}
\eta = \phi_N \circ g_N \dots \circ \phi_1 \circ g_1,
\end{align}
where $\circ$ denotes the function composition. For $1 \le i \le N$, $d_0=1$ and $d_N=3$, $g_i: \R^{d_{i-1}} \to \R^{d_i}$ are composable affine maps given by
\begin{align}
g_i(\tau) = w_i \tau + b_i,
&& w_i \in \R^{d_{i-1} \times d_i},
&& b_i \in \R^{d_i},
\end{align}
and $\phi_i: \R^{d_i} \to \R^{d_i}$ are some activation function. Let $w = \{w_i, b_i\}_{i=1}^{N}$ gather the neural network weights, the neural network can be represented by $\eta(\cdot; w)$. The calibration task with mean squared error \eqref{eq:obj_p} is thus translated to
\begin{align}
\min_{\{w: \ \eta(\cdot; w) \in \Xi\}} L_P(\sigma^\eta, \alpha^\eta, \beta^\eta; w),
\end{align}
that is to train a neural network $\eta$ which represents a term function triplet $(\sigma^\eta, \alpha^\eta, \beta^\eta)$ satisfying Assumption~\ref{ass:term}, and minimizes the calibration error. Moreover, these assumptions can be translated into either soft constraints or hard constraints as follows.

Monotonic conditions \ref{ass:one} and \ref{ass:two} can be facilitated by soft constraints with loss functions
\begin{align} \label{eq:obj_c}
L_C(\sigma^\eta, \alpha^\eta, \beta^\eta; w) 
    = \(\pd_\tau\frac{\alpha^\eta(\tau)}{\sigma^\eta(\tau)}\)^+
    + \(\pd_\tau\frac{\beta^\eta(\tau)}{\sigma^\eta(\tau)}\)^+
    + \(-\pd_\tau \sigma^\eta(\tau)\)^+,
\end{align}
where the derivatives on the neural network can be easily evaluated with automatic differentiation, which is available in existing machine learning frameworks such as \cite{abadi2016tensorflow} and \cite{paszke2017automatic}. The positivity condition \ref{ass:three} is guaranteed by choosing a proper activation function at the final layer. For instance, we can let the last activation function $\phi^N$ in Equation~\eqref{eq:nn} be the softplus activation function. The inequality \ref{ass:four} can be satisfied by reformulating $\eta = (\sigma, \alpha, \beta-\sigma)$, i.e., learning the difference between $\beta$ and $\sigma$ and ensuring its positivity by a certain activation function, instead of learning $\beta$ directly.

Therefore, combining the pricing loss $L_P$ and regularization loss $L_C$, the neural term structure can be trained with the total loss function
\begin{align} \label{eq:obj_total}
L(w) = L_P(\sigma^\eta, \alpha^\eta, \beta^\eta; w) + \lambda L_C(\sigma^\eta, \alpha^\eta, \beta^\eta; w)    
\end{align}
for some regularization parameter $\lambda > 0$ that is properly chosen to penalize any violation of Assumption~\ref{ass:term}. The neural network with the total loss function \eqref{eq:obj_total} can be trained with choices of gradient-based optimizers such as the conventional Adam \cite{kingma2014adam}.

\subsection{Calibration with Sequence of Surfaces} \label{subsec:cali_seq}
Due to perpetually changing market conditions, the implied volatility surface changes dynamically over time as well. Consequently, a previously well-calibrated model might become outdated and cannot produce rightful option prices in the next time step. Recalibration is a common practice in real scenarios to update the model to the newest market information.

A conventional choice is to periodically redo the calibration task \eqref{eq:obj_p} in Section~\ref{subsec:cali_one}. To notate the recalibration procedure clearly in the dynamic scenario, we reserve the lowercase $t$ to the calendar date, let uppercase $T$ notate for the expiry date of options, and denote the tenor of options by $\tau$, e.g., at a time $t$, an option expire at time $T$ has the time-to-maturity $\tau = T - t$. It is practical to quote the market price of a call option by moneyness $\kappa$ and tenor $\tau$ instead of strike $K$ and maturity date $T$, and it is easy to adapt the pricing formula \eqref{eq:option} to such reparametrization. With slight abuse of notations, we override $\cK$ and $\cT$ to be the set of $m$ moneyness and $n$ tenors, i.e. $\cK = \{\kappa_1, \dots, \kappa_m\}$ and $\cT = \{\tau_1, \dots, \tau_n\}$. 

The term structure calibrated to the market at $t$ is then subscribed by the time index as well: $\tau \mapsto (\sigma_t(\tau), \alpha_t(\tau), \beta_t(\tau))$. Assume the calibration happens daily and is scheduled in a certain set of $l \in \N$ time $\cD = \{t_1, \dots, t_l\}$. Based on the discussed recalibration practice, every day we will do optimization \eqref{eq:obj_p} for the optimal term structure, so we will have $l$ sets of term structure $(\sigma_{t_i}(\tau), \alpha_{t_i}(\tau), \beta_{t_i}(\tau))_{i=1}^l$. However, such day-by-day recalibration is not robust.

Even when the market condition mildly changes from $t_1$ to $t_2$, due to numerical instability, the term structure calibrated at $t_2$ might significantly differ from the one calibrated at $t_1$. Consequently, a decision such as hedging based on $t_1$ calibration may deviate substantially from the optima at $t_2$. Therefore, a desirable calibration should not only fit the market at both $t_1$ and $t_2$, but also be robust against self-variation. Such robustness can be enhanced by training the neural term structure jointly with a sequence of historical implied volatility surfaces.

We hence extend the neural term structure in Equation~\eqref{eq:eta} to
\begin{align} \label{eq:eta_t}
\eta: \R_+^2 \to \R_+^3,
&& (t, \tau) \mapsto \eta(t, \tau) = \(\sigma^\eta(t, \tau), \alpha^\eta(t, \tau), \beta^\eta(t, \tau)\).
\end{align}
Correspondingly, the loss function in \eqref{eq:obj_p} is modified to cover the market data across multiple dates in $\cD$:
\begin{align} \label{eq:obj_pt}
L_P'(\sigma, \alpha, \beta)
= \frac{1}{\abs{\cD}\abs{\cK}\abs{\cT}}\sum_{\substack{t\in \cD \\ \kappa \in \cK \\ \tau \in \cT}} \(C^M_t(\kappa, \tau) - C(\kappa, \tau; \sigma(t, \tau), \alpha(t, \tau), \beta(t, \tau))\)^2.
\end{align}
Substituting $L_P$ by $L_P'$ in the total loss function \eqref{eq:obj_total}, we readily adapt the neural term structure framework to deal with the joint calibration task.

By extending the neural term structure to jointly calibrate to the implied volatility surface sequence, in addition to the enhanced robustness of the term structure, we obtain a side product in the form of the \textit{term structure surfaces} represented by $(\sigma(t, \tau), \alpha(t, \tau), \allowbreak \beta(t, \tau))$, which gives a dynamic view of the implied volatility evolution.

Recall that a conventional implied volatility surface $\Sigma_t(\kappa, \tau)$ only informs a static view of the market at time $t$. Along the direction of moneyness $\kappa$, the volatility smile essentially implies that the risk-neutral measure has a heavier tail than the Gaussian density, as used in the Black-Scholes model. The volatility smirk implies the asymmetry of the risk-neutral measure. Both the smile and smirk are quantified by the skew of the implied volatility $\pd_\kappa \Sigma$. Through the direction of tenor $\tau$, the implied volatility term structure $\pd_\tau \Sigma$ informs dispersion of the risk measure from the viewpoint of time-average instantaneous Gaussian volatility.

Though one can quickly infer the current market status from the implied volatility surface, it is hard to keep track of market variation by contrasting multiple implied surface $(\Sigma_{t_i})_{i=1}^{l}$ snapshots across different days. Because the term structure $(\sigma(t, \tau), \alpha(t, \tau), \beta(t, \tau))$ innately encompasses the skews without the explicit parameterization for the moneyness, we can make use of the spare one dimension for the calendar date to reflect the term structure dynamics. Three 3D surfaces with value, time, and tenor axes that can equivalently encode the market information, the term structure surfaces, can be hitherto plotted.

For a fixed date $t$, along the tenor dimension, $\sigma(t, \cdot)$ represents the market view on the dispersion of the risk-neural density through different maturities, while $\alpha(t, \cdot)$ and $\beta(t, \cdot)$ indicates the left-and right-skewness of the density. For a fixed tenor $\tau$, one can make use of $(\sigma(\cdot, \tau), \alpha(\cdot, \tau), \beta(\cdot, \tau))$ to track the market morph trough time. Additionally, for an option released on date $t_0$ with tenor $\tau_0$, then surfaces along the line $(t_0 + \delta t, \tau_0 - \delta t)$ for $0 \le \delta t \le \tau_0$ cover the historical information through the life span of this option.

\section{Numerical Experiments} \label{sec:exp}
We report three experiments in this section. The first experiment uses synthetic data generated with pseudo term structure selected from the parametric functions \eqref{eq:term_func} to validate that the neural term structure \eqref{eq:eta} can perfectly recover the selected functions. In the second experiment, we compare the calibration performance on a single implied volatility surface of the neural term structure \eqref{eq:eta} against the chosen parametric term structure in \eqref{eq:term_func}, showing that the neural term structure outperforms. Finally, we implement the extended version \eqref{eq:eta_t} to fit the historical sequence of implied volatility surfaces and visualize the calibrated term structure surfaces.

\begin{figure}[h]
    \centering
    \includegraphics[width=0.6\linewidth]{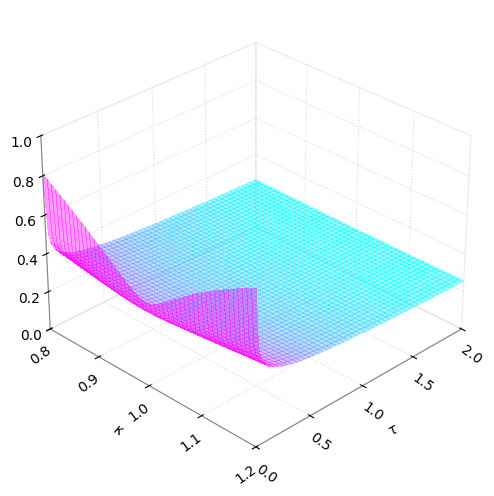}
    \caption{Synthetic implied volatility surface}
    \label{fig:iv_sim}
\end{figure}

\subsection{Synthetic Surface} \label{subsec:exp_syn}

We generate the synthetic surface induced by the additive logistic process using the term structure functions provided in \cite{azzone2023explicit}. Specifically, we choose $(\sigma^\theta, \alpha^\theta, \beta^\theta)$ to be $(\sigma^a, \alpha^b, \beta^b)$, as in Equation~\eqref{eq:term_func}, and let $\theta$ collect those term structure functions as
\begin{equation} \label{eq:theta}
\begin{aligned}
\theta(\tau)
    := \theta(\tau; \sigma_0, \alpha_0, \alpha_1, \beta_0, \beta_1, H_0) 
    = \(\sigma^\theta(\tau; \sigma_0, H_0), \alpha^\theta(\tau; \alpha_0, \alpha_1), \beta^\theta(\tau; \beta_0, \beta_1, \sigma_0, H_0)\).
\end{aligned}
\end{equation}
Here, the parameters of $\theta$ are specified as
\begin{align} \label{eq:theta_param}
\sigma_0 = 0.15, 
&& H_0 = 0.45, 
&& \alpha_0 = 0.8, 
&& \alpha_1 = 1, 
&& \beta_0 = 0.7, 
&& \beta_1 = 2.
\end{align}
The implied volatility surface is gridded on $m=50$ equally spaced moneyness from $0.8$ to $1.2$ and $n=100$ tenors ranging from $1/365.25$ to $2$ that correspond to one day to two years. Figure~\ref{fig:iv_sim} illustrates the synthetic implied volatility surface generated by the aforementioned configuration. The green curves in Figure~\ref{fig:term_sim} visualize the selected parametric functions in blue curves.

The neural term structure $\eta$ is represented by a two-layer feedforward neural network of 32 hidden neurons, each with a ReLU activation function. We use the conventional Adam as the optimizer to minimize the loss function $L$ in Equation~\eqref{eq:obj_total}. The neural term structure with such simple configurations effortlessly converges to the selected parametric functions and meets the conditions in Assumption~\ref{ass:term} with zero loss from the penalty term $L_C$. Figure~\ref{fig:term_sim} plots the learned neural term structure $\eta$ in green curves and we see that they overlap well with the ground truth parametric term structure $\theta$ in blue curves. 

\begin{figure}[h]
    \centering
    \includegraphics[width=0.6\linewidth]{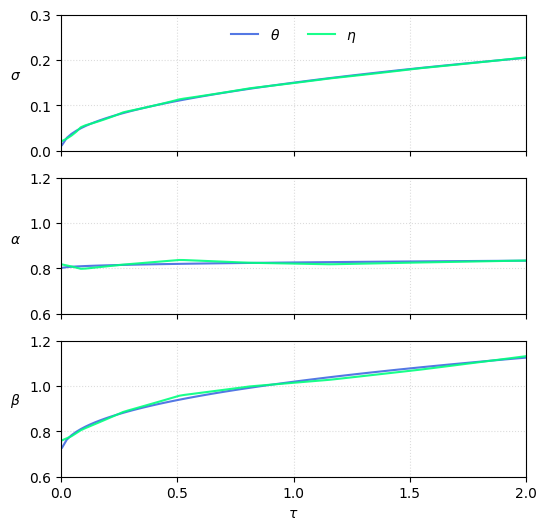}
    \caption{Synthetic term structure}
    \label{fig:term_sim}
\end{figure}

This simple test demonstrates the advantage of the neural term structure. One can easily harness the functional approximation flexibility of neural networks to bypass the hassle of designing proper parametric functions for the term structure.

\subsection{Single Surface} \label{subsec:exp_one}
As explained in Section~\ref{subsec:cali_one}, by prescribing parametric functions, one can only cover a small part of the feasible functional space $\Xi$, which might not capture the true term structure of the real data. Instead, the neural term structure overcomes this limitation, as it encompasses a larger family of functions. This can be seen in the following experiment of single surface calibration.

We choose the S\&P 500 implied volatility surface on 2022-06-30, quoted from Bloomberg, which is described by options in the grid of 13 moneynesses from 0.6 to 1.75 and 43 maturities from one week to approximately two and a half years.

The neural term structure consists of two hidden layers with 256 neurons and ReLU activation in each. The benchmark model is the parametric term structure $\theta$ as in Equation~\eqref{eq:theta}. $\theta$ is optimized against the loss function that is adapted from \eqref{eq:obj_pt} by replacing $\eta$ and its weights $w$ by $\theta$ and the parameters $(\sigma_0, H_0, \alpha_0, \alpha_1, \beta_0, \beta_1)$.

Figure~\ref{fig:iv_hist_one} visualizes the calibration results tenor by tenor. The red, blue, and green curves correspondingly represent the real market implied volatility smiles, the fitted smiles with the parametric term structure $\theta$, and the learned smiles with the neural term structure $\eta$. We see that both the blue curves and the green curves mostly overlap with the red curves, while the green curves align better. That is, both the $\theta$ and $\eta$ fit the implied volatility, yet $\eta$ calibrates better than $\theta$. The performance difference can be seen clearly in Table~\ref{tb:cali_one}.

\begin{figure}[h]
    \centering
    \includegraphics[width=0.9 \linewidth]{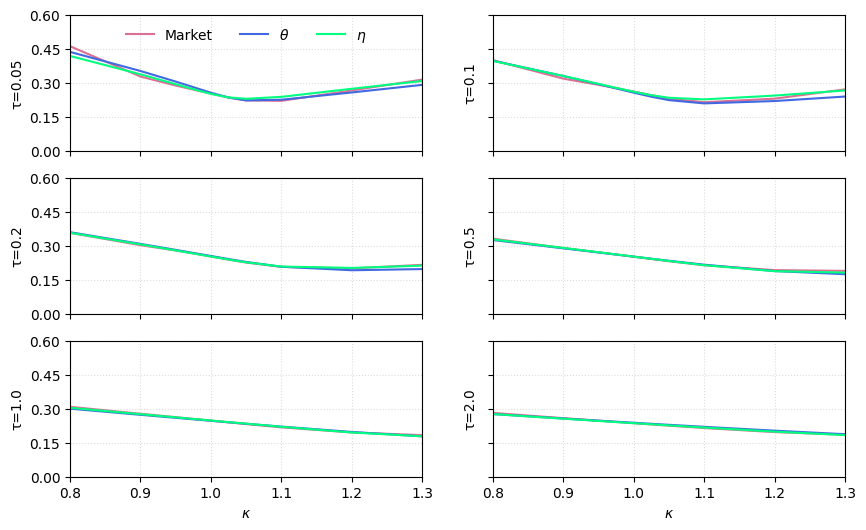}
    \caption{Implied volatility smiles on 2022-06-30}
    \label{fig:iv_hist_one}
\end{figure}

Table~\ref{tb:cali_one} reports both tenor-wise mean squared errors for $\tau \in \{0.05, 0.1, 0.25, 0.5, 1.0, 2.0\}$, which corresponds to about half month, one month, a quarter, half year, one year, and two years, and the overall mean squared calibration error, of $\theta$ and $\eta$. We see that $\eta$ systematically outperforms $\theta$ in either single tenor slice or overall joint tenors.

\begin{table}[h]
\centering
\begin{tabular}{@{}llllllll@{}}
\toprule
          $\tau$ & $0.05$ & $0.1$ & $0.25$ & $0.5$ & $1.0$ & $2.0$ & All \\ \midrule
$\theta$     &   0.45         &   0.33          &    0.17        &    0.68        &    2.06      & 5.72   &  1.59\\
$\eta$ &    0.05        &     0.13        &      0.09      &      0.49      &       1.13  & 2.69 & 0.71\\ 
\bottomrule
\end{tabular}
\caption{Mean squared calibration error $(\times 10^{-2})$}
\label{tb:cali_one}
\end{table}

We examine the difference between the calibrated term structures of $\eta$ and $\theta$ with Figure~\ref{fig:term_hist_one}. The blue curves and the green curves correspondingly represents $(\sigma^\theta, \alpha^\theta, \beta^\theta)$ and $(\sigma^\eta, \alpha^\eta, \beta^\eta)$. Surprisingly, we see the blue curves and the green curves do not agree with each other. This implies that even two divergent term structures may result in closer calibration performance.

\begin{figure}[h]
    \centering
    \includegraphics[width=0.6 \linewidth]{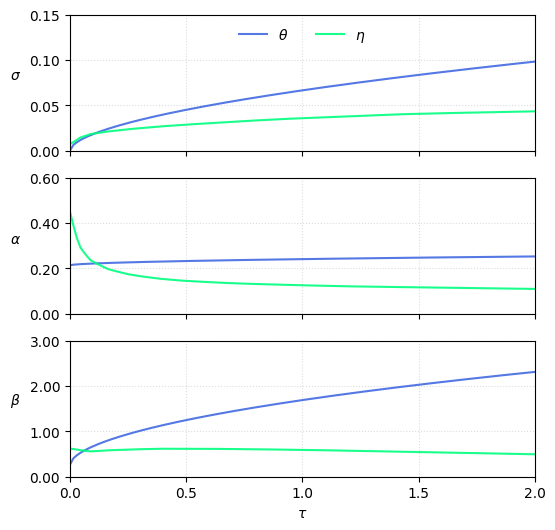}
    \caption{Term structure on 2022-06-30}
    \label{fig:term_hist_one}
\end{figure}

With a closer look, we notice that the GL probability measures produced by both term structures are very similar, so they result in akin implied volatility smiles. This phenomenon raises concerns about the stability of calibration in the practice that one should obtain two dissimilar term structures even on the same data. It would not help with hedging decisions and could be even worse when the implied volatility varies from day to day.   

Therefore, we need a term structure that can not only perform well on one specific surface but also be robust against large variations owning to the surface dynamics. For this, we adapt the neural term structure as in Section~\ref{subsec:cali_seq} to train on a sequence of implied volatility surfaces to produce a set of term structures that have generally consistent shapes and also fit each specific surface.

\subsection{Sequence of Surfaces}
We select the historical sequence of daily S\&P 500 implied volatility surfaces quoted from Bloomberg from 2022-06-30 to 2022-12-30, with calendar time encoded by $t \in [0, 0.5]$. The moneyness-tenor grid is the same as in Section~\ref{subsec:exp_one}.

The extended neural term structure $\eta(t, \tau)$ as in \eqref{eq:eta_t} has a two-layer feedforward neural network, with 256 hidden neurons and ReLU activation functions in each layer. The loss function is as in Equation~\eqref{eq:triplet} with $L_P$ replaced by $L_P'$ in Equation~\eqref{eq:obj_pt}. The Adam optimizer is employed as before and we schedule the learning rate decay to accelerate the training process.

Figure~\ref{fig:iv_hist} demonstrates the learned implied volatility smiles by $\eta$ in green curves versus the market implied volatility smiles in red curves for $\tau=0.1, 0.5$, and $2.0$ tenors. Dates reported are $t=0, 0.28$, and $0.50$ which correspond to the representative dates 2022-06-30, 2022-10-10, and 2022-12-30, where the volatility surge on 2022-10-10 is worth noting. We see that $\eta$ fits the smiles well on every single date, which indicates that it also captures the surface dynamics.

\begin{figure}[h]
    \centering
    \includegraphics[width=0.6\linewidth]{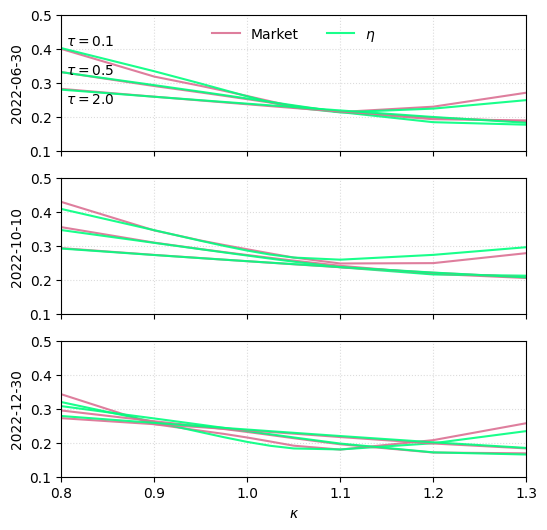}
    \caption{Sequence of implied volatility smiles}
    \label{fig:iv_hist}
\end{figure}

Table~\ref{tb:cali_two} shows the calibration performance of $\eta(t, \tau)$. We also report the results of single surface calibrations $\eta(\tau)$ as in Section~\ref{subsec:exp_one}. The jointly calibrated $\eta(t, \tau)$ has a slightly higher calibration error than the separately calibrated $\eta(\tau)$. However, by surrendering a tiny fraction of accuracy, $\eta(t, \tau)$ attains smooth term structure surfaces that robustly adapt to the changing market conditions.

\begin{table}[h]
\centering
\begin{tabular}{@{}lccc@{}}
\toprule
          Date & 2022-06-30 & 2022-10-10 & 2022-12-30 \\ \midrule
$\eta(t, \tau)$     &   0.98   &   0.57      &    1.00\\
$\eta(\tau)$ &    0.71      &     0.43     &      0.39\\ 
\bottomrule
\end{tabular}
\caption{Mean square calibration error $(\times 10^{-2})$}
\label{tb:cali_two}
\end{table}

Figure~\ref{fig:term_hist} plots these term structure surfaces. We see the general patterns of each surface for $\sigma(t, \tau)$, $\alpha(t, \tau)$, and $\beta(t, \tau)$ are consistent. The sigma surface is increasing and concave through tenor $\tau$. The $\alpha$ and $\beta$ surfaces are mostly decreasing along the tenor. Their evolutions along the calendar axis $t$ reflect the implied volatility dynamics. The bump in the $\sigma$ surface near $t=0.28$ reflects exactly the surge of the overall volatility level around 2022-10-10. The dent in the $\beta$ surface between around $t=0.20$ and $t=0.30$ indicates the decrease of right skewness caused by the at-the-money volatility upturn between roughly 2022-09-09 and 2022-10-17. 

\begin{figure}
    \centering
    \begin{subfigure}{0.6\textwidth}
        \includegraphics[trim={1cm 4cm 1cm 4cm},clip, width=\textwidth]{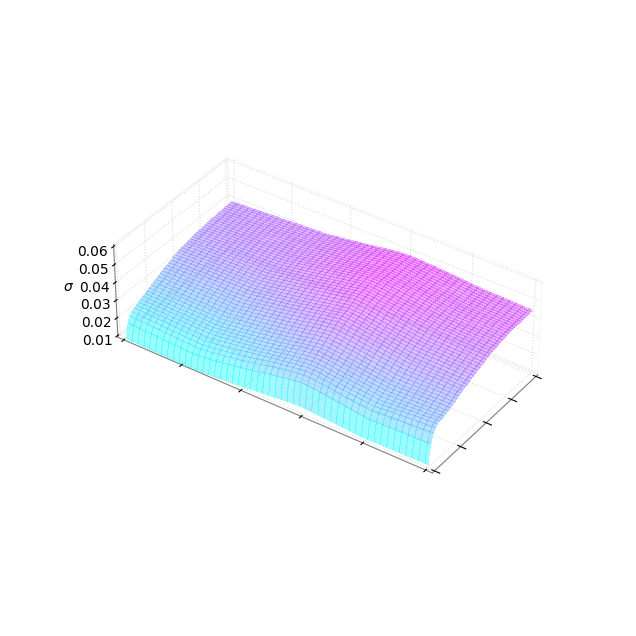}
    \end{subfigure}\\
    \begin{subfigure}{0.62\textwidth}
        \includegraphics[trim={1cm 4cm 1cm 4cm},clip,width=\textwidth]{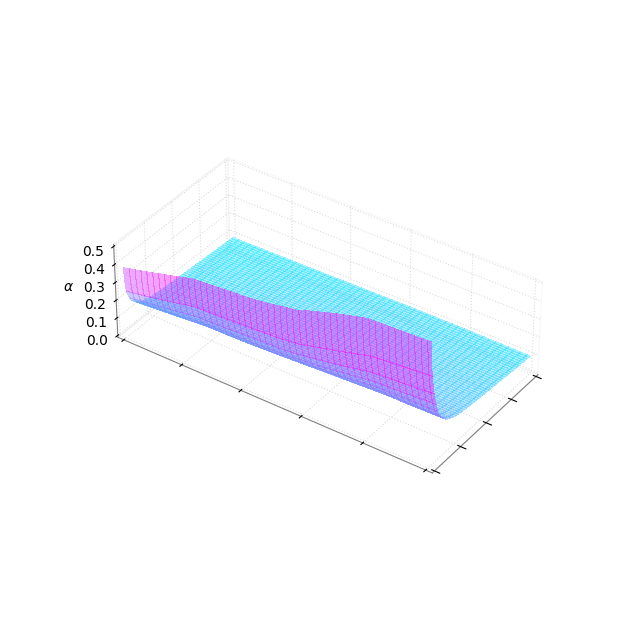}
    \end{subfigure}\\
    \begin{subfigure}{0.6\textwidth}
        \includegraphics[trim={1cm 3cm 1cm 4cm},clip,width=\textwidth]{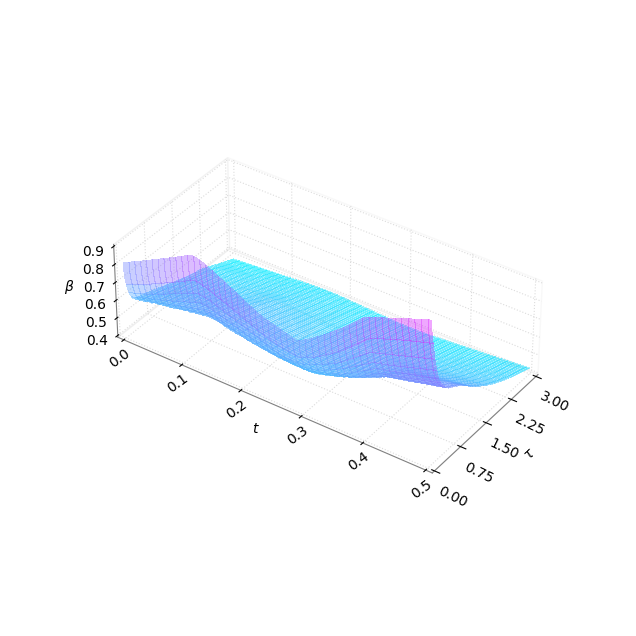}
    \end{subfigure}
    \caption{Calibrated term structure surface}\label{fig:term_hist}
\end{figure}

\section{Conclusion}\label{sec:con}

The neural term structure proposed in this research is a simple, yet powerful approach for calibrating an additive model to implied volatility surfaces. Without the need to prescribe any parametric function, it attains even better calibration accuracy. Moreover, by slightly extending the neural network to feed on time-varying volatility surfaces, it can robustly encapsulate the change of market conditions with smooth flows of term structures that are numerically stable. We demonstrate the advantage of the neural term structure with perhaps the simplest additive process that is supported by the GL distribution. Even as parsimonious as with only three parameters, the additive logistic option pricing model empowered by the neural term structure achieves persuasive calibration performance. It is promising for future research to develop better additive option pricing models by devising more sophisticated additive processes and combining them with the neural term structure.

\bibliographystyle{chicago}
\bibliography{ref}

\end{document}